\documentclass[letterpaper]{article} 
\usepackage{aaai18}  
\usepackage{times}  
\usepackage{helvet}  
\usepackage{courier}  
\usepackage{url}  
\usepackage{graphicx}  
\usepackage{verbatim}

\usepackage{array}
\usepackage{amsmath}
\usepackage{amsthm}
\usepackage{amssymb}
\usepackage{xcolor}

\providecommand{\tabularnewline}{\\}

\theoremstyle{plain}
\newtheorem{thm}{\protect\theoremname}
\theoremstyle{plain}
\newtheorem{prop}{\protect\propositionname}
\theoremstyle{plain}
\newtheorem*{prop*}{\protect\propositionname}
\theoremstyle{plain}
\newtheorem{cor}{\protect\corollaryname}
\theoremstyle{plain}
\newtheorem*{cor*}{\protect\corollaryname}
\theoremstyle{plain}

\theoremstyle{plain}
\newtheorem*{lem*}{\protect\lemmaname}

\providecommand{\corollaryname}{Corollary}
\providecommand{\lemmaname}{Lemma}
\providecommand{\propositionname}{Proposition}
\providecommand{\theoremname}{Theorem}

\nocopyright

\begin{document}

\title{Complexity of Scheduling Charging in the Smart Grid 
}

\author{Mathijs de Weerdt \and Michael Albert \and Vincent Conitzer}
\author{Mathijs de Weerdt\\
Delft University of Technology\\
Mekelweg 4, 2628 CD\\
Delft, The Netherlands\\
M.M.deWeerdt@tudelft.nl
\And
Michael Albert \and Vincent Conitzer\\
Department of Computer Science\\
Duke University\\
Durham, NC 27708, USA\\
\{malbert,conitzer\}@cs.duke.edu
}
\maketitle

\begin{abstract}
In the smart grid, the intent is to use flexibility in demand, both to balance demand and supply as well as to resolve potential congestion. A first prominent example of such flexible demand is the charging of electric vehicles, which do not necessarily need to be charged as soon as they are plugged in. The problem of optimally scheduling the charging demand of electric vehicles within the constraints of the electricity infrastructure is called the {\em charge scheduling problem}. The models of the charging speed, horizon, and charging demand determine the computational complexity of the charge scheduling problem. For about 20 variants, we show, using a dynamic programming approach, that the problem is either in P or weakly NP-hard. We also show that about 10 variants of the problem are strongly NP-hard, presenting a potentially significant obstacle to their use in practical situations of scale.
\end{abstract}

\section{Introduction}

Renewable sources for the generation of electricity are intermittent,
but the amount of power generated needs to equal the amount of power
consumed at all times. Because it is expensive to store electricity
or compensate for the fluctuations with carbon-based generators, there is an incentive for providers to make part of the demand flexible and controllable, i.e., make the (electricity) grid ``smart''. For example, electric vehicle
owners can get a discount on their electricity bill if they allow
their provider to charge their car flexibly~\cite{garcia2014plug}. 
Specifically, car owners may have a deadline by which they would like the vehicle charged (say, by 8:00 in the morning), and they may allow the provider to charge anytime before the deadline. Meanwhile, the supply (bought by the provider) or the network capacity for providing electricity at a given time 
may be limited~\cite{de2015best,philipsen2016imperfect}, requiring providers to intelligently utilize capacity over time.

The problem of deciding when to charge under a common constraint gives rise to a new
class of scheduling problems. The defining difference from the traditional scheduling literature \cite{brucker2007sa,pinedo2012scheduling}
is that such charging jobs are more flexible: not only can they be shifted in time, but the charging speed can also vary over
time. Additionally, the charging resources (``the machines'' in ordinary
scheduling) may vary over time. Further, providers that control flexible demand
will need to solve such scheduling problems repeatedly. Therefore, it is important 
to understand when such problems can be solved optimally within the
time limits required, and what aspects of the model---for example the types of 
user preferences we allow---may make the problem intractable.
We refer to this class of problems generically as the \emph{charge scheduling problem}.


While the existing scheduling literature is extensive~\cite{brucker2007sa,pinedo2012scheduling,hartmann2010survey} and for many scheduling problems the computational complexity is known, the unique setting of charge scheduling gives rise to a number of novel variants of the general scheduling problem.
In this paper over 30 variants are identified, their computational complexity is proven,
and for the easy problems a 
polynomial algorithm is provided.

First, we provide a general model of scheduling flexible demand in
a smart grid. Then we discuss the relation to the traditional scheduling
literature by showing, for various restrictions of the problem, 
equivalence to known scheduling problems. Our main results for classifying
all restrictions of the general model can be found in the section
thereafter.

\section{The Charge Scheduling Problem}

We consider a supply of perishable resources (e.g., network capacity or available power) which varies over time. For concreteness, throughout the paper, we will use the running example of electric vehicles that can be flexibly charged according to their owners' specifications.
We discretize time into intervals $T=\left\{ 1,\ldots,\left|T\right|\right\} $ and use $t\in T$ to indicate a specific interval.
The availability of the resource supply at time interval $t$ is represented by a value $m_{t}\in \mathbb{R}$.
This resource supply is allocated to a set of $n$ agents, and the allocation to agent $i$ is denoted by a function $a_{i}:T\rightarrow\mathbb{R}$.
The value of an agent $i$ for such an allocation is denoted by $v_{i}:[T\rightarrow\mathbb{R}]\rightarrow\mathbb{R}$.
In this paper, we focus on problems where the valuation function of agent $i$ can be represented by triples of a value $v_{i,k}$, a deadline $d_{i,k}$ and a resource demand $w_{i,k}$, such that the value $v_{i,k}$ is  obtained if and only if the demand $w_{i,k}$ is met by the deadline $d_{i,k}$.
This allows the agent to express preferences such as: ``I value being able to go to work at \$100, I must leave for work at 8am, and it requires 25 kWh to complete the trip.''
By adding a second deadline, the agent could express: ``I may suddenly fall ill, so 
I value having at least the option to take my car to urgent hospital care at 10pm at \$20, and it would require 10 kWh to complete that trip.''  If so, the scheduler could make sure that the car is charged up to 10 kWh before 10pm and complete the remaining 15 kWh of charge in the rest of the night, for the full \$120 of value; alternatively, the scheduler may decide that \$20 is too low given others' high evening demands and charge the full 25 kWh later in the night for just \$100 in value (assuming \$100 is large enough).\footnote{Note that if the agent {\em actually} drives to urgent care, then an {\em additional} 25 kWh of charge would be necessary to go to work the next morning, for a total of 35 kWh.  The assumption here is that the probability of actually having to drive to urgent care (and then still wanting to go to work the next morning) is negligible; the agent just wants to have the option.  This ``just-in-case'' nature of some agents' desire to have the car charged at least a little bit earlier in the evening seems realistic.  A more sophisticated approach would be to do the full probabilistic modeling and generate contingency plans for charging rather than simple schedules, but this is beyond the scope of this paper. In any case, much of the paper deals with the case where an agent has only one deadline, where this is not an issue.}

To be able to express such a valuation function concisely, we denote the total amount of resources allocated to an agent $i$ up to and including interval $t$ by $\bar{a}_{i}\left(t\right)=\sum_{t'=1}^{t}a_{i}\left(t\right)$.
We then write
$$v_i(a_i,v_{i,k},d_{i,k},w_{i,k}) = 
\begin{cases}
 v_{i,k} & \text{if } \bar{a}_{i}\left(d_{i,k}\right) \geq w_{i,k}\\
 0 & \text{otherwise}
\end{cases}
$$
and $v_i(a_i) = \sum_{k \in K(i)} v_i(a_i,v_{i,k},d_{i,k},w_{i,k})$, where $|K(i)|$ is the number of deadlines for $i$.

When we say (for some variants) that $m_t$, $v_{i,k}$, or $d_{i,k}$ is \emph{polynomially bounded} by the size of the input, we mean that there exists a polynomial function $p(\cdot)$ such that for all $t$, $m_t \leq p(n,|T|)$, or, for all $i,k$, $v_{i,k} \leq p(n,|T|)$, or $d_{i,k} \leq p(n,|T|)$, respectively.
We aim to find an allocation that maximizes social welfare subject
to the resource constraints, i.e., 
\[
\max\sum_{i}v_{i}(a_{i})
\]

\[
\text{subject to }\sum_{i}a_{i}(t)\leq m_{t}\text{ for every }t
\]
The inequality in the constraint implies free disposal, which in many situations, such as network capacity, is realistic.
This problem has $n\cdot T$ decision variables.

In this paper we further consider variants of this problem
along the following dimensions:
\begin{itemize}
\item Each agent $i$ has a maximum \emph{charging speed $s_{i}$} and for all $t$ and $i$, $a_i(t) \leq s_i$. We consider three variants of such a constraint, namely fixed / unbounded / gaps: 
\emph{fixed} means that the maximum
charging speed is the same at all times, 
\emph{unbounded} means that there
is no bound on the charging speed for each individual agent, and
\emph{gaps} means that the maximum charging speed may
be 0 for some time steps and unbounded for others.
\item The number of \emph{periods} $T$ may be constant or polynomially bounded: \emph{constant}
means that there is an a-priori known number of periods for all instances of the problem, denoted by $O(1)$, while \emph{polynomially bounded} means that the number of periods may be large, but is bounded by a polynomial function of the input size, denoted by $O\left(n^{c}\right)$.
\item The model of the \emph{demand} $d_{i,k}$ may be one of constant / polynomial / unbounded, where \emph{constant}
 means that $d_{i,k}\leq D$ for all $i,k$, and that this $D$ is an a-priori known constant, \emph{polynomial} means that each $d_{i,k}$ is bounded by a polynomial function of the input size, and unbounded means that there is no bound on the demand size.
\item We can have either a \emph{single} \emph{deadline} per agent, $k= 1$, or \emph{multiple} deadlines
 where there may be more than one value-demand-deadline triple
which combine into a total value in the way discussed earlier.
In the case of $k=1$ we simply write $v_{i},d_i,w_i$ to denote $v_{i,1},d_{i,1},w_{i,1}$.
\end{itemize}

\section{Relation to Known Scheduling Problems}
Some of the problem variants within the space identified above map exactly to known scheduling problems.
These relations to known results and the consequences are discussed in this section.

We use the three-field notation $\alpha | \beta | \gamma$,  well established in the scheduling literature, to indicate specific known problems.
This notation can be read as follows (note that this summary is limited to the variants relevant for this paper).
The first field is used to indicate the machines that can be used: $\alpha$ is replaced by $1$ for just a single machine (sequencing problem), by $P$ for identical machines in parallel, by $Q$ for machines with different speeds, and in the case where the number of machines, $m$, is given, $\alpha$ is replaced by $Pm$ or $Qm$, respectively.
The second field, $\beta$, is used for processing characteristics and constraints, and it is replaced by a subset of the following: $pmnt$ denotes that jobs can be pre-empted, i.e., interrupted and continued later, $r_j$ denotes that jobs have release times, and $p_i = p$ denotes that the processing times of all jobs are identical.
The third field is used to indicate the objective function: $\sum v_{j}C_{j}$ for minimizing the weighted sum of completion times, and $\sum v_{j}U_{j}$ for minimizing the weighted sum of values of jobs that have not been scheduled before their deadlines.
We are exclusively concerned with the objective function $\sum v_{j}U_{j}$ in this work.

An important generalization of such scheduling problems allows for including resource constraints as well, i.e., the so-called resource-constrained project scheduling problem (RCPSP)~\cite{hartmann2010survey}.
For this model, there are algorithms for extensions to deal with continuously divisible resources~\cite[Ch.12.3]{blazewicz2007handbook}, a varying availability of resources with time~\cite{klein2000project}, and the possibility to schedule subactivities of the same activity in parallel, called fast tracking~\cite{vanhoucke2008impact}.
When all these extensions are considered simultaneously, the charge scheduling problem can be seen as a special case by modeling the available power by one available resource type, and having no conditions on the number of processors.
However, for RCPSP and its extensions, very few complexity results have been published after an initial investigation, which shows that for one resource, without preemption, the problem of minimizing the make-span is strongly NP-hard even if there are only three processors available~\cite{blazewicz1983scheduling}.
This NP-hardness result, however, does not apply to the charge scheduling problem, mainly because of the difference in the objective, but also because in charge scheduling pre-emption is allowed.

Below, we first show the equivalence of the single-deadline charge scheduling problem with unit charging speed and unit supply of resources 
to $1|pmtn|\sum v_{j}U_{j}$.
This variant is in P when values are assumed to be polynomially bounded.
A further consequence of the equivalence is that when the demand of all agents is the same, the problem is equivalent to $1|pmtn;p_{i}=p|\sum v_{j}U_{j}$.
\begin{prop*}
\label{prop:unit-speed-unit-supply}The single-deadline charge scheduling problem with unit
charging speed and unit supply is equivalent to $1|pmtn|\sum v_{j}U_{j}$.
\end{prop*}
\begin{proof} 
Consider the following reduction.  Let an instance of $1|pmtn|\sum v_{j}U_{j}$ be given: for each job
$i$, we have a 
deadline $d_{i}$, processing time $p_{i}$, and a
value $v_{i}$. The objective of this problem is to schedule the subset of jobs
on the single machine (with preemption) before their deadlines such
that the sum of values of unscheduled jobs is minimized. Define now
a charge scheduling problem with supply $m_t=1$ for all $t$, and for each job $i$ an agent $i$
with a charging speed of 1, a deadline $d_{i}$, weight $w_{i}=p_{i}$,
and value $v_{i}$. Then, a solution to the 
scheduling
problem 
containing tasks $J$ and having
cost $\sum_{j\not\in J}v_{j}$ corresponds to a 
charging schedule with value of $\sum_{j\in J}v_{j}$, and vice versa.  (Note the reduction works in the other direction as well.)
\end{proof}
\noindent With this proposition, and because an algorithm exists for $1|pmtn,r_{j}|\sum v_{j}U_{j}$ that runs in time $O\left(n\sum v_{j}\right)$~\cite{lawler1990dynamic}:
\begin{cor*}
The single-deadline charge scheduling problem with unit charging speed and unit supply can be solved in time
$O\left(n\sum v_{j}\right)$, so is in P if values are polynomial,
and weakly NP-complete otherwise.
\end{cor*}

Additionally, the same reduction can be used when all demand is the same, i.e., when $w_i=w$ (so $p_i=p$).
\begin{cor*}
\label{prop:unit-speed-supply-same-demand}The single-deadline charge scheduling problem with
same demand, unit charging speed, and unit supply is equivalent to
$1|pmtn;p_{i}=p|\sum v_{j}U_{j}$.
\end{cor*}
When the demand for each agent is identical, polynomial algorithms are known even if values are not bounded~\cite{baptiste1999polynomial}.
\begin{cor*}
The single-deadline charge scheduling problem with identical demand, unit charging speed, and unit supply can be solved in time $O\left(n^{7}\right)$.
\end{cor*}

For non-unit (but fixed, $m_t=m$) supply, the charge scheduling problem is equivalent to a multi-machine problem where all agents have an identical charging speed $s_i = s$ and the supply is a multiple of the charging speed.
\begin{prop*}
\label{prop:same-speed}The single-deadline charge scheduling problem with identical charging
speed $s$ among agents, and fixed supply $m_t=m=ks$ for some $k\in\mathbb{N}$ is
equivalent to $P|pmtn|\sum v_{j}U_{j}$.
\end{prop*}
\begin{proof}
Let a single-deadline charge scheduling problem instance with fixed supply $m_t=m=ks$ and a set of tasks
$i$ with charge demand $w_{i}$, charging speed $s$, value $v_{i}$,
and deadline $d_{i}$ be given. Observe that this is equivalent
to a charge scheduling problem where supply $m$, charging speed $s$, and demand
$w_{i}$ are all scaled by dividing by $s$, so supply is $m'=k$,
speed is $s'=1$, and demand is $w'_{i}=\frac{w_{i}}{s}$. 
From this we conclude that such a charge scheduling problem with unit charging speed and integer
supply $k$ is equivalent to a scheduling problem with $m'$ machines.
\end{proof}
Using this proposition and known results from the scheduling literature we immediately obtain a dynamic program (DP), and the complexity class that this variant belongs to.
\begin{cor*}
The single-deadline charge scheduling problem with the same charging speed $s$ and fixed
supply $m_t=m=ks$ for $k\in\mathbb{N}$ is weakly NP-hard and has an $O(n^{2}\sum_{i}v_{i})$
DP algorithm, implying that the problem is in P if values are polynomial. 
\end{cor*}
\begin{proof}
This charge scheduling problem is equivalent to $P|pmtn|\sum v_{j}U_{j}$ by the proposition above, 
which in its turn can be modeled
as $Qm|pmtn|\sum v_{j}U_{j}$~\cite{brucker2007sa}. 
For this problem a $O(n^{2}\sum_{i}v_{i})$
DP exists~\cite{lawler1989preemptive}, which is polynomial if 
$v_{i}$ are polynomially bounded. 
Furthermore, $P|pmtn|\sum U_{j}$ is weakly
NP-hard~\cite{lawler1983recent} 
and thus so are $P|pmtn|\sum v_{j}U_{j}$ and the charge scheduling problem.
\end{proof}
Furthermore, this corollary can be generalized to include charge scheduling problems where $\frac{m}{s}$ is not integer.
In this case, we translate directly to $Qm|pmtn|\sum v_{j}U_{j}$ by having $\left\lfloor \frac{m}{s}\right\rfloor $
machines with speed 1, and 1 machine with speed $\frac{m\textrm{ mod } s}{s}$.
\begin{cor*}
The single-deadline charge scheduling problem with identical (fixed) charging speed $s$ and fixed
supply $m_t=m=ks$ for $k\in\mathbb{Q}$ is weakly NP-hard and has an $O(n^{2}\sum_{i}v_{i})$
DP algorithm, making the problem in P if values are polynomial.
\end{cor*}

Concluding, with fixed charging speeds and supply, the charge scheduling problem is equivalent to some known machine scheduling problems and, fortunately, often the easier ones (i.e., the weakly NP-hard variants).
However, when supply (i.e., the number of machines) varies over time, or when charging speeds (i.e., the maximum number of machines allowed for a single job) differ per agent (job), the existing literature does not readily provide an answer to the question of the charge scheduling problem complexity.

\section{Complexity of Charge Scheduling}

\begin{table*}[ht]
\noindent \begin{flushleft}
\begin{tabular}{ccc}
\emph{gaps} & \emph{fixed charging speed} & \emph{unbounded charging speed}\tabularnewline
\begin{tabular}{c|>{\centering}p{1.15cm}|>{\centering}p{1.15cm}|>{\centering}p{1.15cm}|}
\multicolumn{1}{c}{\emph{$|T|$}} & \multicolumn{1}{>{\centering}p{1.15cm}}{demand constant} & \multicolumn{1}{>{\centering}p{1.15cm}}{demand polynomial} & \multicolumn{1}{>{\centering}p{1.15cm}}{demand unbounded}\tabularnewline
\cline{2-4} 
$O(1)$ & P$^{T\ref{thm:c-periods}}$ & P$^{T\ref{thm:c-periods}}$ & weak NP-c$^{T\ref{thm:c-periods}}$\tabularnewline
\cline{2-4} 
$O\left(n^{c}\right)$ & strong NP-c$^{T\ref{thm:strong-gaps}}$ & strong NP-c$^{T\ref{thm:strong-gaps}}$ & strong NP-c$^{T\ref{thm:strong-gaps}}$\tabularnewline
\cline{2-4} 
\end{tabular} & %
\begin{tabular}{c|>{\centering}p{1.15cm}|>{\centering}p{1.15cm}|>{\centering}p{1.15cm}|}
\multicolumn{1}{c}{} & \multicolumn{1}{>{\centering}p{1.15cm}}{demand constant} & \multicolumn{1}{>{\centering}p{1.15cm}}{demand polynomial} & \multicolumn{1}{>{\centering}p{1.15cm}}{demand unbounded}\tabularnewline
\cline{2-4} 
 & P$^{T\ref{thm:c-periods}}$ & P$^{T\ref{thm:c-periods}}$ & weak NP-c$^{T\ref{thm:c-periods}}$\tabularnewline
\cline{2-4} 
 & ? & ? & ? NP-c$^{P\ref{prop:knapsack}}$\tabularnewline
\cline{2-4} 
\end{tabular} & %
\begin{tabular}{c|>{\centering}p{1.15cm}|>{\centering}p{1.15cm}|>{\centering}p{1.15cm}|}
\multicolumn{1}{c}{} & \multicolumn{1}{>{\centering}p{1.15cm}}{demand constant} & \multicolumn{1}{>{\centering}p{1.15cm}}{demand polynomial} & \multicolumn{1}{>{\centering}p{1.15cm}}{demand unbounded}\tabularnewline
\cline{2-4} 
 & P$^{T\ref{thm:c-periods}}$ & P$^{T\ref{thm:c-periods}}$ & weak NP-c$^{T\ref{thm:c-periods}}$\tabularnewline
\cline{2-4} 
 & P$^{T\ref{thm:p-periods}}$ & P$^{T\ref{thm:p-periods}}$ & weak NP-c$^{T\ref{thm:p-periods}}$\tabularnewline
\cline{2-4} 
\end{tabular}\tabularnewline
\end{tabular}
\par\end{flushleft}

\smallskip{}

\caption{\label{tab:Problem-complexity-of}Problem complexity of variants with
a single deadline per agent, where T$x$ refers to Theorem~$x$, P$x$ to Proposition~$x$, and `? NP-c' means that the problem variant is NP-complete, but
it is an open problem whether this is strong or weak. 
}
\end{table*}

We next analyze the complexity of the charge scheduling problem, again with \emph{single deadlines},
depending on gaps / no-gaps, constant or unbounded periods, and fixed or unbounded demand.
This setting encompasses the case where a user of an electric vehicle needs a certain amount of charge for the next day's driving before he or she leaves for work in the morning (single deadline), while allowing for potentially unavailable charging times---e.g., the vehicle cannot be charged while the agent is at work.
Please see the overview in Table~\ref{tab:Problem-complexity-of}.

\begin{prop}\label{prop:NPc}
(The decision version of) the charge scheduling problem is in NP for all variants.
\end{prop}
\begin{proof}
The verification of a schedule can in all cases
be done in polynomial time.
\end{proof}

By a reduction from the knapsack problem, we argue that single-deadline charge scheduling is weakly NP-hard, even for a single period.
\begin{prop}\label{prop:knapsack}
The single-deadline charge scheduling problem with unbounded
demand is (weakly) NP-hard, even when $|T|=1$.
\end{prop}
\begin{proof}
This proof uses a reduction from knapsack, which is (weakly) NP-hard~\cite{Garey79}.
Let a knapsack problem with capacity $W$ and a set of items be given. For each
item with value $v_{i}$ and weight $w_{i}$, create an agent with
a demand $w_{i}$ at deadline $1$ and value $v_{i}$, and no bound on speed $s_i$.
Define the
supply $m_{1}=W$. 
A set of items fits in the knapsack if and only
if the respective agents can be charged before the deadline.
Because there is no bound on the speed, the result holds for unbounded charging speed, but also for fixed and gaps.
\end{proof}

We next consider single-deadline charge scheduling for multiple periods $T$, a supply per period $t$ of $m_t \leq M$, $n$ agents, and demand (and/or charging speed) per agent of at most $L$.
The optimal solution for this problem is denoted by $OPT(m_{1},m_{\text{2}},\ldots,m_{\left|T\right|},n)$, and this is defined by the following recursive function that returns the best we can do with the first $i$ agents only.
{\small
\begin{multline}
\label{alg:dp-constant-period}
OPT(m_{1},m_{\text{2}},\ldots,m_{\left|T\right|},i) =\\
\begin{cases}
0 & \text{if }i=0\\
\max\left\{ OPT(m_{1},m_{\text{2}},\ldots,m_{\left|T\right|},i-1),o\right\}  & \text{otherwise}
\end{cases}
\end{multline}
\begin{multline*}
\text{where } o=\\
\max_{a_{1},\ldotp,a_{\left|T\right|}}\left\{ OPT(m_{1}-a_{1},\ldots,m_{\left|T\right|}-a_{\left|T\right|},i-1)+v_{i}\left(a\right)\right\} .
\end{multline*}
}
In this formulation, gaps and charging speed limits can be incorporated with additional constraints on the possible allocation $a$.
A dynamic programming implementation of this recursive function gives an algorithm that solves this problem in polynomial time if both maximum supply and maximum demand are polynomially bounded and the number of periods is constant.

\begin{prop}\label{prop:dp-constant-period}
The single-deadline charge scheduling problem can be solved in  $O\left(n\cdot M^{|T|}\right)$
space and $O\left(n\cdot M^{|T|}\cdot L^{|T|}\right)$ time where $M\leq n\cdot L$
is the maximum supply, and $L$ is the maximum demand (or maximum charging speed). 
\end{prop}
\begin{proof}
The DP based on the recursion in Equation~\ref{alg:dp-constant-period} requires $O\left(n\cdot M^{|T|}\right)$
space and $O\left(n\cdot M^{|T|}\cdot L^{|T|}\right)$ time.
\end{proof}

In the gaps and fixed charging speed problem variants, the maximum charging speed may be very large, 
and therefore the reduction from knapsack (Proposition~\ref{prop:knapsack}) applies.
However, when the maximum charging speed is polynomially bounded, then also the number of alternatives $a_t$ for a single period $t$ in the $\max$ in the recursive formulation is polynomially bounded, and so is the maximum supply. Therefore, in that case, and with a constant number of periods, the problem is in P.

\begin{thm}\label{thm:c-periods}
With a constant number of periods, the single-deadline charge scheduling problem is weakly NP-complete.
If furthermore the demand, supply, or maximum charging speed is polynomially bounded, 
the charge scheduling problem is in P.
\end{thm}
\begin{proof}
This follows from the run time bound given in Proposition~\ref{prop:dp-constant-period}, which is exponential (only)
in the number of periods and the $\log$ of the total demand (or effective demand, due to charging speed limitations)
and supply.  Note that if either demand or supply is polynomially bounded then effectively the other is as well.
Thus, for constant $|T|$ and polynomially bounded demand, supply, or maximum charging speed, this problem is in P.
When these are not polynomially bounded, weak NP-hardness  follows from Proposition~\ref{prop:knapsack}.
\end{proof}


Next, we let go of the constant number of periods, and provide an algorithm for single-deadline scheduling for a polynomially bounded number of periods, but only for the case where there is no bound on the charging speed, i.e., a charging task can complete in one period if sufficient supply is available.

Consider the following algorithm.
\begin{enumerate}
\item Sort all charging task triples on deadline (increasing, with arbitrary tie-breaking). 
\item Let $M_{1},M_{2},\ldots,M_{n}$ be the \emph{cumulative supply} at the deadlines
of tasks $1,2,\ldots,n$---that is, $M_i = \sum_{t=1}^{d_i} m_t$---and let $M_{0}=0$. 
\item Run a DP based on the following recursion (where $m$ denotes the remaining cumulative supply available for the first $i$ tasks):
{\small 
\begin{multline}\label{alg:dp-p-periods}
OPT(m,i) =\\
\begin{cases}
0 & \text{if }i=0\\
OPT\left(\min\left\{ m,M_{i-1}\right\} ,i-1\right) & \text{if }m<w_{i}\\
\max\left\{ OPT\left(\min\left\{ m,M_{i-1}\right\} ,i-1\right),\right.&\\
\quad \left.v_{i}+OPT\left(\min\left\{ m-w_{i},M_{i-1}\right\} ,i-1\right)\right\}  & \text{otherwise}
\end{cases}
\end{multline}
}
\noindent where the first call is $OPT(M_{n},n)$. 
\item Recover the set of tasks that get allocated and match this to resources
to find a concrete possible allocation.
\end{enumerate}
This DP is similar to the standard one for knapsack in the special case in which all deadlines are equal. 

\begin{prop}\label{prop:p-periods}
The single-deadline charge scheduling problem with unbounded charging speed can be solved in $O\left(n\cdot M_n\right)$
space and time.
\end{prop}
\begin{proof}
The main contribution to the run time comes from the dynamic program (Equation~\ref{alg:dp-p-periods}).
This algorithm requires $O\left(n\cdot M_n\right)$ space and time.
\end{proof}

\begin{thm}\label{thm:p-periods}
The single-deadline charge scheduling problem with unbounded charging speed is weakly NP-complete. 
If furthermore the demand or supply are polynomially bounded, the charge scheduling problem is in P.
\end{thm}
\begin{proof}
If supply is polynomially bounded, the run time bound from Proposition~\ref{prop:p-periods} is polynomial, and so the problem is in P.
This also holds in case demand is polynomially bounded, because all supply above the demand can be ignored.
Otherwise, the algorithm is pseudo-polynomial and weak NP-hardness follows from Proposition~\ref{prop:knapsack}.
\end{proof}


\begin{thm}\label{thm:strong-gaps}
The single-deadline charge scheduling problem with unbounded periods and gaps is strongly NP-complete, even with constant demand and all deadlines at the end ($T$).
\end{thm}
\begin{proof}
This proof is based on the following reduction from exact cover by 3-sets.
Let a set $X$, with $|X|=3q$ and $q\in\mathbb{N}$, and a collection $C$ of 3-element subsets
of $X$ be given. Define the following charge scheduling problem: let $|T|=3q$
and the supply be $1$ per time period. For each $c_{i}\in C$ with
$c_{i}=\left\{ x_{1},x_{2},x_{3}\right\} $, define an agent with a value of $v_i=1$, a deadline $d_i=T$, and demand $w_i=3$, who can only charge during 
times $x_{1}$, $x_{2}$ and $x_{3}$ (by making all other time slot gaps, i.e., setting the charging speed to 0 at those slots).
It is possible to attain an objective value of $q$ if and only if
$C$ contains $q$ non-overlapping
subsets, for the following reasons.  If it contains $q$ such subsets, we can satisfy the corresponding agents' demands for an objective value of $q$.  Conversely, to obtain $q$ objective value, we need to satisfy $q$ agents, who must correspond to nonoverlapping subsets for us to be able to simultaneously satisfy them.
Since exact cover by 3-sets is strongly NP-hard and the reduction is polynomial,
the charge scheduling problem with unbounded periods, gaps, and constant demand is also strongly NP-hard.
This trivially extends to polynomially bounded and unbounded demand.
\end{proof}

The above results are summarized in Table~\ref{tab:Problem-complexity-of}.
Furthermore, when supply is polynomially bounded, the variants with unbounded demand that are weakly NP-complete attain membership in P.
Only for the variants with non-constant numbers of periods $T$ and fixed charging speed is the complexity still open.
The dynamic program based on Equation~\ref{alg:dp-constant-period} can be used for these cases, but has exponential run time $\Omega(M^{|T|})$; 
the dynamic program based on Equation~\ref{alg:dp-p-periods}, however, does not apply, because it schedules taking only supply constraints into account, ignoring any charging speed constraints.
Conversely, the hardness result from Theorem~\ref{thm:strong-gaps} relies on being able to set the charging speed to zero in selected periods, which is not possible in the variant with fixed charging speed.


\subsection{Multiple Deadlines}
\begin{table*}[ht]
\noindent \begin{flushleft}
\begin{tabular}{ccc}
\emph{gaps} & \emph{fixed charging speed} & \emph{unbounded charging speed}\tabularnewline
\begin{tabular}{c|>{\centering}p{1.15cm}|>{\centering}p{1.15cm}|>{\centering}p{1.15cm}|}
\multicolumn{1}{c}{\emph{$|T|$}} & \multicolumn{1}{>{\centering}p{1.15cm}}{demand constant} & \multicolumn{1}{>{\centering}p{1.15cm}}{demand polynomial} & \multicolumn{1}{>{\centering}p{1.15cm}}{demand unbounded}\tabularnewline
\cline{2-4} 
$O(1)$ & P$^{C\ref{cor:multiple}}$ & P$^{C\ref{cor:multiple}}$ & weak NP-c$^{C\ref{cor:multiple}}$\tabularnewline
\cline{2-4} 
$O\left(n^{c}\right)$ & strong NP-c$^{T\ref{thm:multiple-strong}}$ & strong NP-c$^{T\ref{thm:multiple-strong}}$ & strong NP-c$^{T\ref{thm:multiple-strong}}$\tabularnewline
\cline{2-4} 
\end{tabular} & %
\begin{tabular}{c|>{\centering}p{1.15cm}|>{\centering}p{1.15cm}|>{\centering}p{1.15cm}|}
\multicolumn{1}{c}{} & \multicolumn{1}{>{\centering}p{1.15cm}}{demand constant} & \multicolumn{1}{>{\centering}p{1.15cm}}{demand polynomial} & \multicolumn{1}{>{\centering}p{1.15cm}}{demand unbounded}\tabularnewline
\cline{2-4} 
 & P$^{C\ref{cor:multiple}}$ & P$^{C\ref{cor:multiple}}$ & weak NP-c$^{C\ref{cor:multiple}}$\tabularnewline
\cline{2-4} 
 & strong NP-c$^{T\ref{thm:multiple-strong}}$ & strong NP-c$^{T\ref{thm:multiple-strong}}$ & strong NP-c$^{T\ref{thm:multiple-strong}}$\tabularnewline
\cline{2-4} 
\end{tabular} & %
\begin{tabular}{c|>{\centering}p{1.15cm}|>{\centering}p{1.15cm}|>{\centering}p{1.15cm}|}
\multicolumn{1}{c}{} & \multicolumn{1}{>{\centering}p{1.15cm}}{demand constant} & \multicolumn{1}{>{\centering}p{1.15cm}}{demand polynomial} & \multicolumn{1}{>{\centering}p{1.15cm}}{demand unbounded}\tabularnewline
\cline{2-4} 
 & P$^{C\ref{cor:multiple}}$ & P$^{C\ref{cor:multiple}}$ & weak NP-c$^{C\ref{cor:multiple}}$\tabularnewline
\cline{2-4} 
 & strong NP-c$^{T\ref{thm:multiple-strong}}$ & strong NP-c$^{T\ref{thm:multiple-strong}}$ & strong NP-c$^{T\ref{thm:multiple-strong}}$\tabularnewline
\cline{2-4} 
\end{tabular}\tabularnewline
\end{tabular}
\par\end{flushleft}

\smallskip{}

\caption{\label{tab:Problem-complexity-of-1}Problem complexity of variants
with multiple deadlines per agent, where T$x$ refers to Theorem~$x$ and C$x$ refers to Corollary~$x$.}
\end{table*}

Up to this point, we have only allowed an agent to have  a single charging deadline.
However, an agent might like to have some charge available earlier, just in case, resulting in multiple deadlines, as discussed earlier.

In this section, we consider the variants with multiple deadlines; see Table~\ref{tab:Problem-complexity-of-1}. Obviously, this is a strictly harder setting then the single-deadline case.

\begin{prop} \label{prop:multiple}
Any problem variant with multiple deadlines is at least as hard 
as the corresponding variant with a single deadline.
\end{prop}
\begin{proof} 
This follows directly from a (trivial) reduction from the variant with single deadlines to the
variant with multiple deadlines.
\end{proof}

In fact, any of the problem variants with more than two deadlines and a non-constant number of periods is strongly NP-hard, which we show by a reduction from exact cover by 3-sets.
\begin{thm}\label{thm:multiple-strong}
The charge scheduling problem with multiple deadlines and polynomially bounded periods is strongly NP-hard, even with three deadlines per agent, constant demand, and no bound on charging speeds.
\end{thm}
\begin{proof}
This proof is based on the following reduction from exact cover by 3-sets.

Let a set $X$, with $|X|=3q$, and a collection $C$ of 3-element
subsets of $X$ be given. Assume w.l.o.g.~that the elements in $X$
are $\left\{ 1,2,\ldots3q\right\}$. Define the following charging
problem: let $|T|=3q$ and the supply be $1$ per time period. For
each $c_{i}\in C$ with $c_{i}=\left\{ x_{1},x_{2},x_{3}\right\} \subset X$,
define a valuation function $v_{i}$ such that a value of $3q-x_{1}+1$
is obtained if a charge of 1 takes place before time $x_{1}$, an additional
$3q-x_{2}+2$ if an additional charge of 1 takes place before time
$x_{2}$, and an additional value of $3q-x_{3}+3$ if an additional
charge of 1 takes place before time $x_{3}$. 
Then observe the following: 
\begin{lem*}
Any feasible schedule has a value of at most $\frac{9}{2}q^{2}+\frac{9}{2}q$, and this value is attained if and only if there are $q$ agents that have all three of their deadlines met, each just in time (with the charge arriving exactly at the deadline).
\end{lem*}
\begin{proof}
For a slot $i$ to contribute value, it needs to contribute
to a deadline $x_{k}$ with $i\leq x_{k}$. Letting $k \in \{1,2,3\}$ denote whether it is the corresponding agent's first, second, or third deadline, holding $k$ fixed, the maximum value that $i$ can contribute
is if $x_k$ is in $\arg\max_{x}\left\{ 3q-x+k\mid i\leq x\right\} = \{i\}$, for a value of $3q-i+k$. That is, ideally, every slot is used just in time for a deadline.

Furthermore, focusing on optimizing the $k$ terms, there can be at most $q$ slots that are used for a deadline with $k=3$, because
for each of these there must be one slot used for a deadline with 
$k=2$ and one for a deadline with $k=1$.
Similarly, there can be at most $2q$ slots that are used for deadlines with $k=3$ or $k=2$.  Hence, ideally, there are $q$ agents that have all three of their deadlines met.

All schedules thus have a value of at most $\sum_{i=1}^{3q}\left(3q-i\right)+\sum_{i=1}^{q}\left(1+2+3\right)=9q^{2}-\frac{1}{2}3q(3q+1)+6q=\frac{9}{2}q^{2}+\frac{9}{2}q$, and this value is attained only under the conditions of the lemma.
\end{proof}

\noindent To continue our reduction, we show that the optimal charging schedule has value $\frac{9}{2}q^{2}+\frac{9}{2}q$
if and only if $C$ contains $q$ non-overlapping subsets.
\begin{enumerate}
\item If $C$ contains $q$ non-overlapping subsets $\left\{ x_{1},x_{2},x_{3}\right\} $,
then for each of these, the respective agent's charges can be feasibly
scheduled exactly in slots $\left\{ x_{1},x_{2},x_{3}\right\} $,
leading to a value of $\left(3q-x_{1}+1\right)+\left(3q-x_{2}+2\right)+\left(3q-x_{3}+3\right)$
and thus a total value of $\sum_{i=1}^{3q}\left(3q-i\right)+\sum_{i=1}^{q}\left(1+2+3\right)=\frac{9}{2}q^{2}+\frac{9}{2}q$,
which is optimal (according to the above lemma).
\item If the optimal charging schedule has value $\frac{9}{2}q^{2}+\frac{9}{2}q$,
then this can only be because $q$ agents have been allocated 
three slots each, all exactly at their respective deadlines, according to the above lemma.
Because all slots are
allocated at the respective deadlines, and there is only one deadline
per slot, the sets of deadlines are not overlapping, and hence the $q$ agents whose deadlines are met correspond to an exact cover.
\end{enumerate}
Since the reduction is polynomial and exact cover by 3-sets is strongly NP-hard, so is the charge scheduling problem with multiple deadlines,  polynomially bounded periods, constant demand and no bound on charging speeds.
This also directly implies strong NP-hardness in the case where demand is polynomial or unbounded and in the case where charging speeds may have gaps or are fixed.
\end{proof}

While Theorem~\ref{thm:multiple-strong} indicates that (assuming P$\neq$NP) it is impossible to solve the charge scheduling problem with multiple deadlines over an arbitrary horizon in polynomial time, in many cases, it may be sufficient to examine only a relatively short period.
Specifically, a grid manager may only be able to optimize over a single upcoming day due to uncertainty in longer term power production and uncertainty in consumers' preferences over longer horizons.
The following corollary indicates that in this setting, solving the charge scheduling problem may still prove feasible in practice.

\begin{cor} \label{cor:multiple}
With a constant number of periods, the charge scheduling problem with multiple deadlines is weakly NP-complete.
If furthermore the demand, supply, or maximum charging speed is polynomially bounded, then the problem is in P.
\end{cor}
\begin{proof}
The dynamic program based on Equation~\ref{alg:dp-constant-period} also works with multiple deadlines:
the value function $v_i(a)$ used is the general one from the original problem definition.
Since the run time is exponential (only) in the number of periods
and the $\log$ of the total demand (or effective demand, due to charging speed limitations) and supply (see Proposition~\ref{prop:dp-constant-period}), it follows that when the number
of periods is constant, it is pseudo-polynomial, and the problem is weakly NP-complete by Proposition~\ref{prop:multiple}. If furthermore any of the total demand, supply, or maximum charging speed is polynomial, then (effectively) so are the others, and the problem is polynomial-time solvable (i.e., in P).
\end{proof}
With this result, we have established the computational complexity of all problem variants with multiple deadlines, as can be seen in Table~\ref{tab:Problem-complexity-of-1}.

\section{Discussion and Future Work}

The detailed analysis of the complexity of charge scheduling and the dynamic programs for some of the variants provide an important step towards practical applicability.
However, two aspects of high importance deserve a more in-depth discussion: 
first, in most real-life situations we will need to deal with the arrival of new agents, a so-called on-line problem~\cite{albers2009online}.
Second, the approach should also work in case agents are self-interested and have the possibility to state their preferences strategically.
Both these extensions are discussed below, including the relation between our contributions and relevant known results.

\subsection{On-line Charge Scheduling}

The \emph{on-line charge scheduling problem} is equivalent to the charge scheduling problem introduced earlier, except that each agent $i$ becomes known to the scheduler only after its arrival/release time $r_i$.
A common approach is to use competitive analysis to measure the performance of an on-line algorithm:
we compare the performance of an on-line algorithm to that of a clairvoyant algorithm that has all relevant information a priori.
We say that the on-line algorithm has a \emph{competitive factor} $r$, $0 < r < 1$, if it is guaranteed to achieve a value of at least $r$ times the value obtained by an optimal clairvoyant algorithm on any input.

\citeauthor{baruah1992competitiveness} (\citeyear{baruah1992competitiveness}) have demonstrated, using an adversary argument, that for a setting equivalent to the single-deadline on-line charge scheduling problem where (i) all charging speeds are equal to the supply, and (ii) all demand has the same value-density (i.e., $\frac{v_{i}}{w_{i}}=c$ for all $i$), no on-line scheduling algorithm can have a competitive factor greater than $1/4$~\cite{baruah1992competitiveness}.
Since the same input could occur with on-line charge scheduling, the following result is immediate.
\begin{cor*}
No on-line single-deadline charge scheduling algorithm can have a competitive factor greater than $1/4$, even with fixed supply, fixed charging speeds, and constant demand.
\end{cor*}

A further bound comes from results on the competitiveness of on-line algorithms for bandwidth allocation~\cite{bar1995bandwidth}.
These results use the maximum ratio of the requested bandwidth to the available bandwidth, denoted by $\delta$.
In on-line charge scheduling, $\delta=\max_{i,t} \frac{s_i}{m_t}$.
If charging is done either at maximum speed or not at all, the following result for fixed supply ($m_t=m_t'$ for all $t,t'$) follows directly from the work by \citeauthor{bar1995bandwidth}.
\begin{cor*}
If realized charging speeds needs to be either the maximum $s_i$ or 0, and $\delta>0.5$, then no constant competitive deterministic algorithm for the on-line charge scheduling problem exists.
\end{cor*}
On the other hand, some positive results exist as well.
For example, there are several works proposing an on-line algorithm for other variants of charge scheduling of electric vehicles, such as without the supply constraint~\cite{tang2014online} or with a weak supply constraint~\cite{yu2016intelligent}.
Also, there are randomized algorithms for other variants of on-line scheduling and  algorithms that are competitive with the best ratio possible~\cite{koren1995d}.
It remains an open question which of these can be effectively applied to the on-line charge scheduling problem defined here.

\subsection{Eliciting User Preferences}
\label{sec:mech-elic-user}

A further additional difficulty in scheduling flexible charging capacity not directly addressed by this work is that of eliciting the valuations, deadlines, and quantities (i.e., $v_{i,k}, d_{i,k}$, and $w_{i,k}$) from the agents using the grid.
If the appropriate incentives are not provided, a rational, self-interested agent is likely to report a very high valuation for the earliest possible deadline with the maximum possible charge that 
is feasible for the grid to deliver.
However, this will lead to sub-optimal allocation of grid capacity, as capacity is allocated based on a selfishly inflated view of the agents' requirements.

We can in principle overcome this issue by eliciting the information from self-interested agents using techniques from \emph{mechanism design}---see, e.g.,~\citeauthor{vulkan2000ems}~(\citeyear{vulkan2000ems}) and~\citeauthor{nisan01}~(\citeyear{nisan01}).
There exists a standard mechanism for efficiently allocating resources, the \emph{ Vickrey-Clarke-Groves (VCG) mechanism} \cite{vickrey1961caa,clarke71a,groves73}.
Intuitively, the VCG mechanism allocates resources efficiently and then charges each agent for the loss to other agents' welfare that results from its presence.
Under the VCG mechanism it is a \emph{dominant strategy} to report \emph{truthfully}.

While the VCG mechanism would directly solve the problem of efficiently and truthfully eliciting agents' private information, the mechanism has a significant downside relevant for this line of research; the mechanism must compute the efficient allocation of resources over the outcome space.
If an approximate solution is used in the place of the true efficient allocation, in addition to the mechanism no longer efficiently allocating resources, agents can have an incentive to lie \cite{nisan01,lehmann2002tra}.
In this work, we have demonstrated that in many settings of interest it is computationally infeasible to compute the socially efficient allocation, implying that it is computationally infeasible to implement the VCG mechanism.
However, if the VCG mechanism  instead always computes an allocation that is \emph{maximal-in-range}, i.e., it is an optimal allocation among all allocations that the algorithm will ever select, then the VCG mechanism again truthfully elicits the agents' information \cite{nisan07a}.

As we have demonstrated, if time periods are constant and valuations are polynomially bounded, then the charge scheduling problem is computationally tractable.
Therefore, we propose that a fruitful future line of research would be to examine optimal ways to discretize the allocation space (into a fixed number of time periods and a fixed number of possible valuations and demands reported by the agents) such that the problem is computationally tractable within this restricted (discretized) range.
We leave this exploration to future work.

\subsection{Acknowledgements}
\label{sec:acknowledgements}

Conitzer is thankful for support from NSF under awards IIS-1527434 and CCF-1337215.

\bibliographystyle{aaai}
\bibliography{main,gametheory,scheduling,weerdt,onlinescheduling}

\end{document}